\documentclass[12pt,draftclsnofoot,journal,onecolumn]{IEEEtran}
\IEEEoverridecommandlockouts
\usepackage{mathtools}
\usepackage{hyperref}
\usepackage{ifpdf}
\usepackage{amsthm}
\usepackage{amsmath}
\usepackage{nicefrac}
\usepackage{cite}
\usepackage{amsfonts}
\usepackage{comment}
 \usepackage{url}
 \usepackage{amssymb}
 \usepackage[paper=letterpaper,margin=0.75in]{geometry}
 \usepackage[ansinew]{inputenc}
\theoremstyle{definition}
\newtheorem{thm}{Theorem}

\begin{document}
\title{Covert Single-hop Communication in a Wireless Network with Distributed Artificial Noise Generation}

\author{
   \IEEEauthorblockN{Ramin Soltani\IEEEauthorrefmark{1},
           Boulat Bash\IEEEauthorrefmark{2}, Dennis Goeckel\IEEEauthorrefmark{1}, Saikat Guha\IEEEauthorrefmark{3}, and Don Towsley\IEEEauthorrefmark{2}}

\IEEEauthorblockA{\IEEEauthorrefmark{1}Electrical and Computer Engineering Department, University of Massachusetts, Amherst,
    \{soltani, goeckel\}@ecs.umass.edu}
    \IEEEauthorblockA{\IEEEauthorrefmark{2}School of Computer Science, University of Massachusetts, Amherst,
    \{boulat, towsley\}@cs.umass.edu}
    \IEEEauthorblockA{\IEEEauthorrefmark{3}Raytheon BBN Technologies
        ,\{sguha\}@bbn.com}
        				\thanks{This work has been supported, in part, by the National Science Foundation under grants CNS-1018464 and ECCS-1309573.}
                       \thanks{This work has been presented at the 52nd Annual Allerton Conference on Communication, Control, and Computing, Allerton, Monticello, IL, October 2014. The final version of this work~\cite{soltani2017covert} has been submitted to the IEEE for possible publication.}
                       \thanks{Personal use of this material is permitted. Permission from IEEE must be obtained for all other uses, in any current or future media, including reprinting/republishing this material for advertising or promotional purposes, creating new collective works, for resale or redistribution to servers or lists, or reuse of any copyrighted component of this work in other works. DOI: \href{https://doi.org/10.1109/ALLERTON.2014.7028575}{10.1109/ALLERTON.2014.7028575} }

}

\date{}
\maketitle

\thispagestyle{plain}
\pagestyle{plain}

\begin{abstract}
Covert communication, also known as low probability of detection (LPD) communication, prevents the adversary from knowing that a communication is taking place. Recent work has demonstrated that, in a three-party scenario with a transmitter (Alice), intended recipient (Bob), and adversary (Warden Willie), the maximum number of bits that can be 
transmitted reliably from Alice to Bob without detection by Willie, when additive white Gaussian noise (AWGN) channels exist between all parties, is on the order of the square root of the number of channel uses. In this paper, we begin consideration of network scenarios by studying the case where there are additional ``friendly'' nodes present in the environment that can produce artificial noise to aid in hiding the communication. We establish achievability results by considering constructions where the system node closest to the warden produces 
artificial noise and demonstrate a significant improvement in the throughput achieved covertly, without requiring close coordination between Alice and the noise-generating node. Conversely, under mild restrictions on the communication strategy, we demonstrate no higher covert throughput is possible. Extensions to the consideration of the achievable covert throughput when multiple wardens randomly located in the environment collaborate to attempt detection of the transmitter are also considered.
\end{abstract}
%%%%%%%%%%%%%%%%%%%%%%%%%%%%%%%%%%%%%%%%%%%%%%%%%%%%%%%%%%%%%%%%%%%%%%%%%%%%%%%%%%%%%%%%%%%%%%%%%%%%%%%%%%%

\section{Introduction}

The provisioning of security has emerged as a critical issue in wireless communications to
prevent unauthorized access to the information sent from the transmitter to the desired recipient.
Standard security approaches, whether they are computational (cryptographic) or information-theoretic,
focus on preventing the eavesdropper from obtaining the contents of the message.  However,
it has recently become apparent that a significant threat to users' privacy is not only
the discovery of a message's content, but also the very existence of the message itself
(e.g. the seeking of ``meta-data,'' as detailed in the Snowden disclosures \cite{snowden}).
This motivates the consideration of covert (i.e. low probability of detection) communications.

Historically, covert communication has been of military interest, and spread spectrum approaches
have been widely considered \cite{simon94ssh}.  However, the fundamental limits of covert communication were
only recently established by a subset of the authors \cite{bash_isit2012, bash_jsac2013},
who presented a square root limit on the number of bits that can be transmitted securely from the
transmitter (Alice) to the intended receiver (Bob) when there are additive white Gaussian 
noise (AWGN) channels between Alice and each of Bob and the adversary (Warden Willie).  In particular, by taking 
advantage of the non-zero noise power at Willie, Alice can reliably transmit $\mathcal{O}(\sqrt{n})$ bits 
to Bob over $n$ uses of a channel while lower bounding Willie's error probabilities $\mathbb{P}_{FA}+\mathbb{P}_{MD}\geq 1-\epsilon $ for any $\epsilon>0$ where $\mathbb{P}_{FA}$ is the probability of false alarm and $\mathbb{P}_{MD}$ is the probability of mis-detection. Conversely, if Alice transmits more than $\mathcal{O}(\sqrt{n})$ 
bits over $n$ uses of channel, either Willie detects her with high probability or Bob suffers a non-zero probability of decoding error as $n$ goes to infinity. Covert communications has recently attracted the attention of other researchers \cite{jaggi_isit2013, jaggi_isit2014, hou_kramer_2014} and further work of the authors\cite{bash_isit2013, bash_isit2014}.

In this paper, we turn our attention to the network case, where a collection of 
nodes work to establish covert communication between a collection of source and destination pairs.  
The goal is to establish an analog to the line of work on scalable low probability of intercept
communications \cite{jsac_goeckel, suddu_mobihoc,capar_infocom,capar_ciss}, which considered the 
extension of \cite{gupta_kumar, francheschetti} to the {\em secure} multipair unicast problem in 
large wireless networks.  Here, in analog to \cite{jsac_goeckel}, we consider how security between
Alice and Bob can be improved when there are a number of other nodes present in the environment. Whereas \cite{jsac_goeckel} considered
low probability of intercept (LPI) communications, which allowed pilot signaling for protocol
set-up, the consideration of covert communication is more 
challenging, as we assume that Willie allows no communications from Alice whatsoever.

Consider a wireless network with AWGN channels between Alice and each of Bob and Willie.  The power received at any node is inversely proportional to $d^{\gamma}$, where $d$ is the distance of the receiver from the transmitter and $\gamma$ is the path-loss exponent.  Alice attempts to communicate covertly with Bob without detection by Willie, but also in the presence of other (friendly) network nodes to assist the communication by producing 
background chatter to inhibit Willie's ability to detect Alice's transmission.  We assume the friendly nodes 
are distributed according to a two-dimensional Poisson point process of density 
$m={o}\left(n^{{1}/{\gamma}}\right)$.  Alice and Bob share a secret (codebook) that is unknown to Willie. 
For this scenario, which is described in more detail in Section \ref{prerequisites}, we show in Section
\ref{theorems} that Alice can covertly transmit $\mathcal{O}(m^{\gamma/2} \sqrt{n})$ bits to the receiver Bob, 
who is a unit distance away, over $n$ uses of the channel while keeping Willie's sum of error probabilities 
$ \mathbb{P}_{FA} + \mathbb{P}_{MD} \geq 1- \epsilon $  for any $ \epsilon \geq 0 $, hence demonstrating that the presence of friendly nodes, if sufficiently dense, can significantly improve covert throughput.   Conversely, if Alice attempts to transmit $\omega(m^{\gamma/2} \sqrt{n}) $ bits to Bob over $n$ uses of the channel, either there exists a detector that Willie can use to detect her with arbitrarily low sum of error probabilities $\mathbb{P}_{FA} + \mathbb{P}_{MD}$ or Bob cannot decode the message with arbitrarily low probability of error.  In Section \ref{mult_willies}, the extension to
the case of multiple collaborating Willies located in the field is also presented, which establishes the 
framework for a single transmission on a multi-hop path in a large network. 

\section{Prerequisites}
\label{prerequisites}

\subsection{System Model}

Consider a source Alice wishing to communicate with receiver Bob located at a unit distance away in the presence of adversaries $W_1, W_2, \dots , W_{N_w}$, who are distributed independently and uniformly in the unit square shown in Fig.~\ref{fig:SysMod}  and seek to detect any transmission by Alice. When there is only a single Willie, we omit the subscript and denote it by $W$. Also present are friendly nodes allied with Alice and Bob. These nodes, which are distributed according to a two-dimensional point process with density $m={o}(n^{{1}/{\gamma}})$, where $\gamma$ is the path-loss exponent, are willing to help hide Alice's transmission by generating noise. We assume that the system is able to determine which friendly node is the closest to each Willie. The adversaries try to detect whether Alice is transmitting or not by processing their received signals and applying hypothesis testing on them, as discussed in the next subsection.
We consider three scenarios: single Willie located half way between Alice and Bob, single Willie located randomly and uniformly in the 1 by 1 square shown as a dashed box in Fig.~\ref{fig:SysMod}, and multiple Willies scenario where $N_w$ Willies are located independently and randomly in the unit box. Discrete-time AWGN channels with real-valued symbols are assumed for all channels.  Alice transmits $n$ real-valued symbols $f_1, f_2, ..., f_n$. Each friendly node is either on or off according to the strategy employed. Let $\theta_j$ be one when the $j^{th}$ friendly node is ``on'' (transmits noise) and zero otherwise (silent). If $R_j$ is on, it transmits symbols $\{ f_i^{(j)} \}_{i=1}^{\infty}$, where $\{ f_i^{(j)} \}_{i=1}^{\infty}$ is a collection of independent and identically distributed (i.i.d.) zero-mean Gaussian random variables, each with variance (power) $P_r$.

\begin{figure}
\begin{center}
 \includegraphics[ 
scale=0.6]{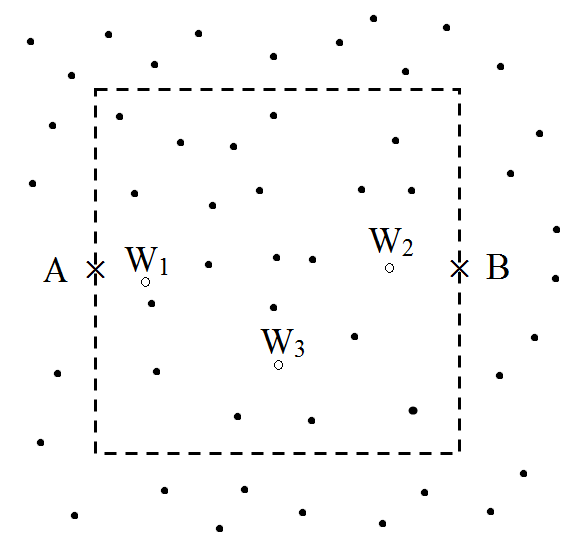}
\end{center}
 \caption{System Configuration: Source node A wishes to communicate reliably and without detection to 
the intended receiver B at distance one (normalized) with the assistance of friendly nodes (represented by solid nodes in the figure) distibuted according to a two-dimensional point process with density $m$  in the presence of adversary nodes  $W_1, W_2, \dots , W_{N_w}$ located in the dashed box ($N_w = 3$ in the figure). }
 \label{fig:SysMod}
 \end{figure}

Bob receives $y_1^{(b)}, y_2^{(b)}, ..., y_n^{(b)} $ where $y_i^{(b)} = f_i+z_i^{(b)}$ for $1\leq i \leq n$. The noise 
component is $z_i^{(b)}= z_i^{(b_0)}+\sum_{j=1}^\infty \theta_j z_j^{(b,r_j)}$, where 
$\{z_i^{(b_0)}\}_{i=1}^{n}$ is an i.i.d 
sequence representing the background noise of Bob's receiver with $z_i^{b_0} \sim \mathcal{N}(0,\sigma_{b_0}^2)$ for all $i$, and $\{z_i^{(b,r_j)}\}_{i=1}^{n}$ is an i.i.d. sequence of received noise samples caused by chatter from the $j^{th}$ friendly node when it is ``on'', 
with $\mathbb{E}[|z_j^{(b,r_j)}|^2]=\frac{P_r}{d_{b,r_j}^\gamma}$, where $d_{x,y}$ is the distance between node
$x$ and node $y$.

% Single Willie Sys Model%
\begin{comment}
In addition, Willie observes  $y_1^{(w)}, y_2^{(w)}, ..., y_n^{(w)}$ where  $y_i^{(w)} = f_i+z_i^{(w)}$. Also, $z_i^{(w)}= z_i^{(w_0)}+\sum_{j=1}^m \theta_j z_j^{w_{r_j}}$ where $z_i^{w_0} \sim  \mathcal{N}  (0,\sigma_{w_0}^2)$ is the i.i.d AWGN of the channel, $z_j^{(w,r_j)}$ is the noise generated by $j^{th}$ friendly node when it is on with variance $\frac{P_r}{d_{(w,r_j)}^\gamma}$. 
\end{comment}

% Multiple Willie SyS Model%

Similarly, the $k^{th}$ Willie ($1 \leq k \leq N_w$) observes  $y_1^{(w_k)}, y_2^{(w_k)}, ..., y_n^{(w_k)}$ where  $y_i^{(w_k)} = f_i+z_i^{(w_k)}$. Here, $z_i^{(w_k)}= z_i^{({w_k}_0)}+\sum_{j=1}^\infty \theta_j z_j^{({w_k},r_j)}$ where $\{z_i^{({w_k}_0)}\}_{i=1}^{n}$ is an i.i.d sequence representing the background noise at Willie's receiver, where $z_i^{({w_k}_0)} \sim \mathcal{N}  (0,\sigma_{{w_k},0}^2)$ for all $i$, and $\{z_i^{({w_k},r_j)}\}_{i=1}^{n}$ is the i.i.d sequence of received noise samples caused by chatter from the $j^{th}$ friendly node when it is "on" with variance $\frac{P_r}{d_{{w_k},r_j}^\gamma}$. 

Note that we assume Alice and the friendly nodes, while having a common goal, are not able to closely align their transmissions; that is, the friendly nodes set up a constant power background chatter but are not able to, for example, lower their power at the time Alice transmits.

\subsection{Hypothesis Testing}

Consider the case of a single Willie. We assume he applies a hypothesis test to his received signal to determine whether or not Alice is communicating with Bob. This test is performed according to Willie's knowledge about his channel to Alice. When Alice is not transmitting, Willie expects to observe Gaussian white noise along with transmissions from other nodes and, when Alice is transmitting, he expects to observe a signal with greater power. We denote the probability distribution of Willie's collection of observations $\{y_i^{({w}_0)}\}_{i=1}^{n}$ by $\mathbb{P}_1$ when Alice is communicating with Bob, and the distribution of the observations when she is not transmitting by $\mathbb{P}_0$.

There are two hypotheses, $H_0$ and $H_1$. The null hypothesis ($H_0$) corresponds to the case that Alice is not transmitting, and the alternative hypothesis $H_1$ corresponds to the case that Alice is transmitting. We denote by $\mathbb{P}_{FA}$ as the probability of rejecting $H_0$ when it is true (type I error or false alarm), and $\mathbb{P}_{MD}$ as the probability of rejecting $H_1$ when it is true (type II error or mis-detection). We assume that Willie uses classical hypothesis testing with equal prior probabilities and seeks to minimize $\mathbb{P}_{FA} + \mathbb{P}_{MD}$; the generalization to arbitrary prior probabilities is straightforward~\cite{bash_jsac2013}. For a scenario with multiple collaborating Willies (Theorem 3), the received signals are processed together at a server to arrive at a single collective decision as to whether Alice is transmitting or not. 

\subsection{Reliability and Covertness}

We define Alice's transmission as reliable if and only if the desired receiver (Bob) can decode her message with arbitrarily low average probability of error $\mathbb{P}_e$ at long block lengths, where the average is over the node locations. In other words, for any $\zeta>0$, Bob can achieve $\mathbb{P}_e < \zeta$ as $n \to \infty$.

Alice's transmission is covert if and only if she can lower bound Willie's (or Willies', for scenarios with multiple adversary nodes) average sum of probabilities  of error ($\mathbb{E}\left[\mathbb{P}_{FA}+\mathbb{P}_{MD}\right]$) by $1-\epsilon$ for any $\epsilon>0$, as $n \to \infty$ ~\cite{bash_jsac2013}.

\section{Covert Communication in the Presence a Single Warden and $m$ Friendly Nodes}

\label{theorems}

In this section, we first consider the case where there is only one Willie located half-way 
between Alice and Bob. To hide the presence of Alice's transmission, we turn on the friendly node closest to Willie and then analyze Willie's ability to detect Alice's transmission. This allows us to derive a restriction on Alice's power required to maintain covertness.  The achievability proof concludes by considering the rate at
which reliable decoding is still possible under this restriction on Alice's power level.  A converse under mild restrictions on the signaling scheme is also provided. After considering the case where Willie is located half-way between Alice and Bob in Theorem 1, we analyze the problem of a single Willie located randomly and uniformly in the 1 by 1 square shown as a dashed box in Fig.~\ref{fig:SysMod}.

% \subsection{Single Willie located at the center}

\begin{thm}

When friendly nodes are distributed such that $m={o}\left(n^{{1}/{\gamma}}\right)$ and $m=\omega{(1)}$ and there is one warden (Willie) located half-way between Alice and Bob, Alice can reliably and covertly transmit $  \mathcal{O}(m^{\gamma/2}\sqrt{n}) $ bits to Bob over $n$ uses of the channel.

Conversely, if Alice attempts to transmit $\omega(m^{\gamma/2} \sqrt{n}) $ bits to Bob over $n$ uses of channel, either there exists a detector that Willie can use to detect her with arbitrarily low sum of error probabilities $\mathbb{P}_{FA} + \mathbb{P}_{MD}$ or Bob cannot decode the message with arbitrarily low probability of error.

\end{thm}

\begin{proof}

{\it (Achievability)} 

\textbf{Construction}: To establish secret communication, Alice and Bob share a codebook that is not revealed to Willie. For each message transmission of length $M$, Alice uses a new codebook to encode the message into a codeword of length $n$ at the rate of $R=\frac{M}{n}$. To build a codebook, random coding arguments are used; that is, codewords $\{C(W_l)\}_{l=1}^{l=2^{nR}}$ are associated with messages $\{W_l\}_{l=1}^{l=2^{nR}}$, where each of the codewords $C(W_l)=\{C^{(u)}(W_l)\}_{u=1}^{u=n}$ includes random symbols $C^{(u)}(W_l) \sim \mathcal{N}(0,P_f)$ where $P_f$ is defined later. At the receiver, Bob employs a maximum-likelihood (ML) decoder to process his received signal. 

To establish a covert communication, Alice and Bob's strategy is to turn on the closest friendly node to Willie and keep all other friendly nodes off, whether Alice is transmitting or not.  Therefore, Willie's observed noise power is given by
\begin{equation}
{\sigma_w^2} = {\sigma_{w_0}^2} + \frac{P_r}{d_{r,w}^{\gamma}} ,
\end{equation}

\noindent where $\sigma_{w_0}^2$ is Willie's noise power when none of the friendly nodes are transmitting and $d_{r,w}$ is the (random) distance of the closest friendly node to Willie; hence, ${\sigma_w^2}$ is a random variable which depends on the locations of the friendly nodes.

\textbf{Analysis}:  When Willie applies the optimal hypothesis test \cite{bash_jsac2013}:
\begin{equation} \label{eq:th1e1} \mathbb{E}_{R}[\mathbb{P}_{FA}+\mathbb{P}_{MD}] \geq 1- \mathbb{E}_{R} \left[    \sqrt{\frac{n}{2} \mathcal{D}(\mathbb{P}_w || \mathbb{P}_s)} \; \right]   
\end{equation}

\noindent where  $\mathbb{E}_R[.]$ denotes the expected value over all possible locations of the friendly nodes,  $\mathcal{D}(\mathbb{P}_w || \mathbb{P}_s)$ is the relative entropy between  $\mathbb{P}_w$ and  $\mathbb{P}_s$, $ \mathbb{P}_w = \mathcal{N}(0,\sigma_w^2) $ is the probability distribution function (pdf) for each of Willie's observations $z_i^{({w}_0)}$ when Alice is not transmitting and $ \mathbb{P}_s = \mathcal{N}(0,\sigma_w^2+\frac{P_f}{d_{w,a}^{\gamma}}) $ is the pdf for each of the corresponding observations when Alice is transmitting.

We next show how Alice can lower bound the sum of average error probabilities by upper bounding  $\mathbb{E}_{R} \left[    \sqrt{\frac{n}{2} \mathcal{D}(\mathbb{P}_w || \mathbb{P}_s)} \; \right] $. For the given $ \mathbb{P}_w$ and  $ \mathbb{P}_s$ we calculate the relative entropy ~\cite{bash_jsac2013}:
\begin{align}\label{eq:thm1d}
\nonumber \mathcal{D}(\mathbb{P}_w || \mathbb{P}_s) & = \int_x p_0(x) \, \ln \frac{p_0(x)}{p_1(x)} dx   \\
&= \frac{1}{2} \left[ \ln \left(1+ \frac{P_f}{d_{w,a}^{\gamma} \sigma_w^2} \right) -\left(1+ \left( \frac{P_f}{d_{w,a}^{\gamma} \sigma_w^2} \right)^{-1} \right)^{-1}  \right].
\end{align}
Suppose Alice sets her average symbol power $P_f \leq \frac{c m^{\gamma/2}}{\sqrt{n}} $ where $c$ is a constant defined later. Since $m$ is ${o}\left({n}^{{1}/{\gamma}}\right)$, for $n$ large enough $P_f \leq 2\sigma_{w_0}^2 d_{w,a}^{\gamma}< 2\sigma_w^2 d_{w,a}^{\gamma}$. Then, using the Taylor series expansion at $P_f=0$ yields 
\begin{equation}\label{eq:thm1d2}
\mathcal{D}(\mathbb{P}_w || \mathbb{P}_s) <  \left(\frac{P_f}{2 d_{w,a}^{\gamma}\sigma_w^2}\right)^2.
\end{equation}

\noindent Since $d_{w,a}=\frac{1}{2} $,
\begin{align}
\nonumber \mathbb{E}_{R} \left[    \sqrt{\frac{n}{2} \mathcal{D}(\mathbb{P}_w || \mathbb{P}_s)} \; \right]   &< {2^{\gamma-1} P_f} \sqrt{\frac{n}{2}} \mathbb{E}_{R} \left[\frac{1}{\sigma_w^2}\right]  \\
&\leq \frac{2^{\gamma-1}}{\sqrt{2}} c m^{\gamma/2}   \mathbb{E}_{R} \left[\frac{1}{\sigma_w^2}\right].
\end{align}

\begin{comment}
The cumulative distribution function (cdf) of $d_{r,w}$ is given by

\begin{align}
\nonumber \mathbb{F}_{d_{r,w}}\left(x\right) = 1-e^{-m \pi x^2},
\end{align}

\end{comment}

%shortenned of this:
\begin{comment}
\begin{align}
\nonumber \mathbb{F}_{d_{r,w}}\left(x\right) &=\mathbb{P}(\min(d_{r_i ,w}) \leq x)\\ & = 1-\mathbb{P}(\min(d_{r_i ,w}) \geq x)\nonumber\\
& \nonumber = 1-\mathbb{P}\left(\bigcap_{i=1}^{m}{\{d_{r_i ,w} \geq x\}} \right)   \\
&= 1-e^{-m \pi x^2},
\end{align}
\end{comment}

\noindent Due to the Poisson assumption, the pdf of $d_{r,w}$ is easily obtained as
\begin{equation}
\label{eq:p1}
{f}_{d_{r,w}}\left( x \right) = 2 m \pi x \, e^{-m \pi x^2 }.
\end{equation}

\noindent Therefore,
\begin{align}\label{eq:totalex1}
\nonumber m^{\gamma/2} \mathbb{E}_{R} \left[\frac{1}{\sigma_w^2}\right]  &=  m^{\gamma/2} \mathbb{E}_{R}\left[\frac{1}{\sigma_{w_0}^2+{P_r}/{d_{r,w}^\gamma}}\right] \\
\nonumber & \leq \frac{m^{\gamma/2}}{P_r} \mathbb{E}_{R}\left[d_{r,w}^\gamma \right]\\
\nonumber & = \frac{2 m^{\gamma/2+1} \pi}{P_r} \int_{x=0}^{x=\infty} x^{\gamma+1} e^{-m \pi x^2 } dx\\
  &= \frac{\Gamma \left(\gamma/2+1\right)}{2 P_r \pi^{\gamma/2+1}} ,
\end{align}

\noindent where $\Gamma(.)$ is the Gamma function. If Alice sets $c \leq \frac{\epsilon \sqrt{2}}{2^{\gamma-1}} \left(\frac{\Gamma \left(\gamma/2+1\right)}{2 P_r \pi^{\gamma/2+1}} \right)^{-1}$, she can achieve $\mathbb{E}_{R} \left[    \sqrt{\frac{n}{2} \mathcal{D}(\mathbb{P}_w || \mathbb{P}_s)} \; \right] <\epsilon$. Thus, with $P_f \leq  \frac{c m^{\gamma/2}}{\sqrt{n}}$, Alice can covertly transmit to Bob. Note that Alice does not use the locations of the friendly nodes to select the transmission power (and thus, per below, the corresponding rate). Rather, she can choose a power and corresponding rate that is covert when averaged over the locations of the friendly nodes.

Now, we analyze Bob's decoding error probability averaged over all possible codewords and locations of friendly nodes. For Bob's ML decoder, the decoding error probability averaged over all possible codewords conditioned on $\sigma_b^2=\sigma_{b_0}^2+\frac{P_r}{d_{r,b}^{\gamma}}$ where $d_{r,b}$ is the distance from Bob to the relay closest to Willie, is upper bounded using (5)-(9) in \cite{bash_jsac2013}:
\begin{align}
\nonumber \label{eq:Pebasic}   \mathbb{P}_e \left(\sigma_b^2\right)  & \leq 2^{nR- \frac{n}{2} \log_2 \left( 1+ \frac{ P_f}{2 \sigma_b^2} \right) }\\
 &=2^{nR- \frac{n}{2} \log_2 \left( 1+ \frac{ c  m^{\gamma/2} }{2 \sqrt{n} \sigma_b^2} \right) }.
 \end{align}
  
 If the rate is set to $R= \frac{\rho}{2} \log_2 \left( 1+ \frac{ c  m^{\gamma/2} }{ 2 \sqrt{n} \left(\sigma_{b_0}^2+{4^{\gamma} P_r }\right)} \right) $, $0 < \rho < 1$,
  \begin{align}\label{eq:th1bobpe}
    \nonumber \mathbb{P}_e \left(\left. \sigma_b^2\right|d_{r,b}>\frac{1}{4}\right)& \leq 2^{-(1-\rho ) \frac{n}{2} \log_2 \left( 1+ \frac{c  m^{\gamma/2} }{ 2 \sqrt{n} \left(\sigma_{b_0}^2+{4^{\gamma} P_r }\right)} \right) }\\
  &  \nonumber =  \left( 1+ \frac{c   m^{\gamma/2} }{2 \sqrt{n} \left(\sigma_{b_0}^2+{4^{\gamma} P_r}\right)} \right)^{-(1-\rho) \frac{n}{2}} \\
  &  \leq  \left({1+\frac{c  m^{\gamma/2} \sqrt{n} (1-\rho)}{4 \left(\sigma_{b_0}^2+{4^{\gamma} P_r}\right)}}\right)^{-1}.
  \end{align}
  
\noindent  where \eqref{eq:th1bobpe} is due to $(1+x)^r \leq \left({1-rx}\right)^{-1}$ for any $r<0$ and $x>0$. The expected value of $\mathbb{P}_e \left(\sigma_b^2\right)$ over all possible values of the distance of the closest friendly node to Willie is:
\begin{align}\label{eq:EPEBob}
\nonumber \mathbb{P}_e= \mathbb{E}_{R}\left[\mathbb{P}_e \left(\sigma_b^2\right)\right]& =\mathbb{E}_{R}\left[\left. \mathbb{P}_e \left(\sigma_b^2\right) \right|d_{r,b}\leq  \frac{1}{4} \right] \mathbb{P}(d_{r,b}\leq  \frac{1}{4})\\
 &\phantom {=} + \mathbb{E}_{R}\left[\left.\mathbb{P}_e \left(\sigma_b^2\right)\right|d_{r,b}>  \frac{1}{4}\right] \mathbb{P}(d_{r,b}>  \frac{1}{4})
\end{align}
\noindent Consider
\begin{align}\label{eq:EPEBob1}
\nonumber \mathbb{E}_{R}\left[\left. \mathbb{P}_e \left(\sigma_b^2\right)\right|d_{r,b}\leq  \frac{1}{4}\right] \mathbb{P}\left(d_{r,b}\leq  \frac{1}{4}\right)
&\nonumber  \leq  \mathbb{P}\left(d_{r,b}\leq  \frac{1}{4}\right)\\
&\nonumber  \leq \mathbb{P}\left(d_{r,w}> \frac{1}{4}\right)\\
& = e^{-\pi m \left(\frac{1}{4}\right)^2}
\end{align}

\noindent Next, consider the term in~\eqref{eq:EPEBob}:
\begin{align} \label{eq:frstpart}
 \nonumber &\mathbb{E}_{R}\left[\left.\mathbb{P}_e\left(\sigma_b^2\right)\right|d_{rb}>  \frac{1}{4}\right] \mathbb{P}\left(d_{rb}>  \frac{1}{4}\right) \\& \leq \mathbb{E}_{R}\left[\left.\mathbb{P}_e\left(\sigma_b^2\right)\right|d_{rb}>  \frac{1}{4}\right]\\ 
&\nonumber \leq \mathbb{E}_{R}\left[\left.\mathbb{P}_e\left(\sigma_b^2\right)\right|d_{rb}=  \frac{1}{4}\right] \\
 & =  \left({1+{ \frac{c (1-\rho ) m^{\gamma/2} \sqrt{n}}{4 \left(\sigma_{b_0}^2 + {4^{\gamma} P_r}\right)}}}\right)^{-1}.
     \end{align}
    
\noindent Thus, by Eqs.~\eqref{eq:EPEBob},~\eqref{eq:EPEBob1},~\eqref{eq:frstpart},  $\lim\limits_{m \to \infty} \mathbb{P}_e =0$ and, for any $0< \zeta < 1 $, $\mathbb{P}_e < \zeta$.

Now, we calculate the average number of bits that Bob can receive. Since $m$ is ${o}\left({n}^{{1}/{\gamma}}\right)$, for $n$ large enough $\frac{c m^{\gamma/2}}{2 \sqrt{n}} < 2\sigma_b^2$. Based on the fact that for any $0<x<1$, $\log_2{(1+x)} \geq {x}$
% upper bound
\begin{comment}
Since $\ln(1+x)\leq {x}$ for any $x$

\begin{align}
\label{eq:thm1bobp25}\mathbb{E}_{R} \nonumber\left[ nR \right]& =
\mathbb{E}_{R} \nonumber\left[ n \frac{\rho}{2 } \ln \left( 1+ \frac{ c   m^{\gamma/2} }{ 2 \sqrt{n} \sigma_b^2} \right) \frac{1}{\ln{2}}\right]\\
\nonumber &\leq  \mathbb{E}_{R} \left[\frac{\sqrt{n} \rho c  m^{\gamma/2} }{4 \sigma_b^2 \ln 2}\right] \\
& \nonumber =  \frac{\sqrt{n} \rho c  m^{\gamma/2} }{ 4  \ln{2} } \mathbb{E}_{R} \left[\frac{1}{\sigma_b^2 }\right]\\
& \leq \frac{\sqrt{n} \rho c m^{\gamma/2} }{ 4 \ln {2}} \frac{1}{ \sigma_{b_0}^2 }.
\end{align}

On the other hand
\end{comment}
\begin{align}
\label{eq:thm1bobp3} nR & \geq \frac{\sqrt{n} \rho c  m^{\gamma/2} }{4 \left(\sigma_{b_0}^2 + {4^{\gamma} P_r}\right)}.
\end{align}
\noindent Thus, Bob receives $\mathcal{O}(m^{\gamma/2} \sqrt{n})$ bits in $n$ channel uses.
\begin{comment}
Now, we calculate the average number of bits that Bob can receive. 
Since for any $\kappa>0$ $ \lim\limits_{m \to \infty} \mathbb{P}\left(\left|d_{{rw}}-0\right|\geq \kappa\right) = \lim\limits_{m \to \infty} e^{-\pi m\kappa^2}=0$, $d_{rw} \xrightarrow{p} 0$. On the other hand, $\frac{1}{2}-d_{rw}\leq d_{rb} \leq \frac{1}{2}+d_{rw}$ yields $d_{rb} \xrightarrow{p} \frac{1}{2}$ and $\sigma_b^2 =\left( \sigma_{b_0}^2+ \frac{P_r}{d_{rb}^{\gamma}}\right)\xrightarrow{p} \sigma_{b_0}^2+ 2^{\gamma} P_r$. Therefore

\begin{align}
\nonumber nR &\xrightarrow{p}  \frac{n\rho}{2} \log_2 \left( 1+ \frac{c \sqrt{2} m^{\gamma/2} }{ \sqrt{n} \left(\sigma_{b_0}^2+ 2^{\gamma} P_r\right)} \right)\\
 \nonumber &\leq  \frac{n \rho}{2} \left( \frac{c f(n) m }{2 \sqrt{n} \left(\sigma_{b_0}^2+ 4 P_r\right)} \right)\\
 \label{eq:th1bobbit}&= \frac{ \rho c\left(  \sqrt{n} f(n) m  \right)}{4\left(\sigma_{b_0}^2+ 4 P_r\right)} 
\end{align}

\noindent where the last step in\eqref{eq:th1bobbit} is due to the fact that  $\log_2(1+x)<x$ for $x>0$. Thus, Bob receives $\mathcal{O}(m\sqrt{n})$ bits in $n$ channel uses.

\end{comment}

{\it (Converse)} Suppose Willie uses a power detector on his collection of observations $\left\{y_i^{(w)}\right\}_{i=1}^{n}$ to form $S=\frac{1}{n} \sum_{i=1}^{n}\left(y_i^{(w)}\right)^2$ and performs a hypothesis test based on the comparison of $S$ to a threshold $t$. When $H_0$ is true \cite{bash_jsac2013}
\begin{align}
\label{eq:th1con1}\mathbb{E}[S]&=\sigma_w^2\\
\label{eq:th1con2}\mathrm{Var}[S]&=\frac{2\sigma_w^4}{n}
\end{align}

\noindent When $H_1$ is true
\begin{align}
\label{eq:th1con3}\mathbb{E}[S]&=\sigma_w^2+P_k\\
\label{eq:th1con4}\mathrm{Var}[S]&=\frac{4 P_k \sigma_w^2 + 2\sigma_w^4}{n}
\end{align}

\noindent where $P_k$ is the power of the codeword sent by Alice. If $S<\sigma_w^2+t$, Willie accepts $H_0$; otherwise, he accepts $H_1$.  Bounding $\mathbb{P}_{FA}$ by using Chebyshev's inequality yields \cite{bash_jsac2013}:
\begin{align}
\mathbb{P}_{FA} \leq \frac{2 \sigma_w^4}{n t^2}
\end{align}

\noindent Therefore
\begin{align}
\nonumber \mathbb{E}_R [\mathbb{P}_{FA}] &= \mathbb{E}_R  \left[\left.\mathbb{P}_{FA}\right|d_{r,w}\leq \eta_1 \right] \mathbb{P}(d_{r,w}\leq \eta_1 )\\
  \nonumber  &\phantom{=}+ \mathbb{E}_R \left[\left.\mathbb{P}_{FA}\right|d_{r,w}>\eta_1 \right] \mathbb{P}(d_{r,w}>\eta_1 )\\
 \nonumber &\leq \mathbb{P}(d_{r,w}\leq\eta_1 ) +  \mathbb{E}_R \left[\left.\mathbb{P}_{FA}\right|d_{r,w}>\eta_1 \right]\\
 &\leq  \left(1-e^{- m \pi\eta_1^2}\right) + \frac{2 \left(\sigma_{w_0}^2 + \frac{P_r}{\eta_1^\gamma}\right)^2}{n t^2}
\end{align}

\noindent $\forall \eta_1>0$. Let Willie choose threshold $t = \frac{2\sqrt{2}}{\sqrt{n \lambda}} \left(\sigma_{w_0}^2 + \frac{ P_r}{\eta_1^{\gamma}}\right)$ where $\eta_1 < \sqrt{\frac{\ln {\left(\frac{4}{4-\lambda}\right)}}{m \pi}}$. Then
\begin{align}
 \mathbb{E}_R [\mathbb{P}_{FA}] <  \left(1-\left(1-\frac{\pi}{4}\right)\right) + \frac{2 n \lambda}{8 n} = \frac{\lambda}{2}
\end{align}
\noindent In addition, Willie can upper bound $\mathbb{P}_{MD}$ by (16) in~\cite{bash_jsac2013}
\begin{align} \label{eq:th1conv1}
\mathbb{P}_{MD} \leq \frac{4 P_k \sigma_w^2 + 2 \sigma_w^4}{n \left(P_k-t\right)^2}
\end{align}

\noindent Therefore
\begin{align}
\nonumber \mathbb{E}_R [\mathbb{P}_{MD}] &= \mathbb{E}_R \left[\left.\mathbb{P}_{MD}\right|d_{rw}\leq\eta_2 \right] \mathbb{P}(d_{rw}\leq\eta_2 )\\
  \nonumber  &\phantom{=}+ \mathbb{E}_R \left[\left.\mathbb{P}_{MD}\right|d_{rw}>\eta_2 \right] \mathbb{P}(d_{rw}>\eta_2 )\\
 \nonumber &\leq \mathbb{P}(d_{rw}\leq\eta_2 ) +  \mathbb{E}_R \left[\left.\mathbb{P}_{MD}\right|d_{rw}>\eta_2 \right]\\
 &\nonumber \leq  \left(1-e^{- m \pi\eta_2^2}\right) + \frac{4 P_k \left(\sigma_{w_0}^2 + \frac{P_r}{\eta_2^{{\gamma}}}\right)}{n \left(P_k-t\right)^2}\\
 & \phantom{=} +  \frac{2 \left(\sigma_{w_0}^2 + \frac{P_r}{\eta_2^{\gamma}}\right)^2}{n \left(P_k-t\right)^2}
 \end{align}

\noindent $\forall \eta_2>0$. We now set $\eta_2 = \sqrt{\frac{\ln {\left(\frac{2}{2-\lambda+\lambda^{\prime}}\right)}}{m \pi}}$, where $0<\lambda^{\prime}<\lambda$. Since $m$ is ${o}({n}^{{1}/{\gamma}})$ and $t$ is $\Theta\left(\frac{m^{\gamma/2}}{\sqrt{n}}\right)$, if Alice sets her average symbol power $P_k = \omega\left(\frac{m^{\gamma/2}}{\sqrt{n}}\right)$, then there exists $n_0>0$ s.t. $\forall n>n_0(\lambda^{\prime})$
\begin{align}\nonumber \mathbb{E}_R [\mathbb{P}_{MD}] \leq \frac{\lambda-\lambda^{\prime}}{2} + \frac{\lambda^{\prime}}{2} = \frac{\lambda}{2}
\end{align}

\begin{comment}
\begin{align}
\lim\limits_{m,n \to \infty} \frac{4 P_k \left(\sigma_{w_0}^2 + \frac{P_r}{\eta_2^\gamma}\right)}{n \left(P_k-t\right)^2} +  \frac{2 \left(\sigma_{w_0}^2 + \frac{P_r}{\eta_2^\gamma}\right)^2}{n \left(P_k-t\right)^2} = 0
\end{align}
\end{comment}

 \noindent Therefore $\frac{\lambda}{2}$ and $\mathbb{P}_{FA} + \mathbb{P}_{MD} < \lambda$ for any $\lambda>0$.
 
Thus, to avoid detection for a given codeword, Alice must set the power of that codeword to $P_{\mathcal{U}} = O\left(\frac{m^{\gamma/2}}{\sqrt{n}}\right)$. Suppose that Alice's codebook contains a fraction $\xi>0$ of codewords with power $P_{\mathcal{U}} = O\left(\frac{m^{\gamma/2}}{\sqrt{n}}\right)$. Bob's decoding error probability of such low power codewords is lower bounded by (Eq. (20) in \cite{bash_jsac2013})
\begin{align} \label{eq:th1con5}
\mathbb{P}_e^{\mathcal{U}} \geq 1- \frac{\frac{P_{\mathcal{U}}}{2\sigma_b^2}+\frac{1}{n}}{\frac{\log_2 \xi}{n}+R}
\end{align}

\noindent Since Alice's rate is $R=\omega \left(\frac{m^{\gamma/2}}{\sqrt{n}}\right)$ bits/symbol, $\lim\limits_{n,m \to \infty} \mathbb{P}_e^{\mathcal{U}}$ is bounded away from zero. 
\end{proof}
\begin{thm}
When friendly nodes are distributed such that $m={o}\left(n^{{1}/{\gamma}}\right)$ and $m={w}\left(1\right)$, and there is just one warden (Willie) located randomly and uniformly over the unit square shown in Fig.~\ref{fig:SysMod}, Alice can reliably and covertly transmit $  \mathcal{O}(m^{\gamma/2}\sqrt{n}) $ bits to Bob over $n$ uses of the channel.

\begin{comment}
Conversely, if Alice attempts to transmit $\omega(m^{\gamma/2} \sqrt{n}) $ bits to Bob over $n$ uses of the channel, either there exists a detector that Willie can use to detect her with arbitrarily low sum of error probabilities $\mathbb{P}_{FA} + \mathbb{P}_{MD}$ or Bob cannot decode the message with arbitrarily low probability of error.
\end{comment}
\end{thm}

\begin{proof}
  
\textbf{Construction:} We use the same construction and strategy as in Theorem 1.

\textbf{Analysis:} By Eqs.~\eqref{eq:th1e1} and~\eqref{eq:thm1d2}
\begin{align}
\nonumber  & \label{eq:expcond1} \mathbb{E}_{R,W} [ \mathbb{P}_{FA} + \mathbb{P}_{MD}|d_{w,a}>\psi] \\
&  \geq  1- \sqrt{\frac{n}{2}} \mathbb{E}_{R,W}\left[\left.\frac{P_f}{2 \sigma_w^2 d_{w,a}^{\gamma}}\right|d_{w,a}>\psi\right]
\end{align}

\noindent where $\mathbb{E}_{R,W}$ denotes the expectation over all locations of friendly nodes and Willie and $\psi$ is a parameter such that $0 < \psi <\frac{1}{2}$. Suppose Alice sets ${P}_f \leq \frac{c m^{\gamma/2}}{\sqrt{n}} $ where $c$ is a constant defined later. Therefore,
\begin{align} 
&\nonumber  \mathbb{E}_{R,W} \left[\left. \mathbb{P}_{FA} + \mathbb{P}_{MD}\right|d_{w,a}>\psi\right] \\& \nonumber \geq 1-\frac{c}{2 \sqrt{2}} \mathbb{E}_{R,W}\left[\left.\frac{m^{\gamma/2}}{{{\sigma}}_w^2 d_{w,a}^{\gamma}}\right|d_{w,a}>\psi\right]\\
 \label{eq:th2eq1}
& \geq 1-\frac{c }{2 \sqrt{2}\psi^{\gamma}}
\mathbb{E}_{R,W}\left[\left.\frac{m^{\gamma/2}}{{{\sigma}}_w^2}\right|d_{w,a}>\psi\right],
\end{align}

\noindent As in Eq.~\eqref{eq:totalex1}, $\mathbb{E}_{R,W}\left[\left.\frac{m^{\gamma/2}}{{{\sigma}}_w^2}\right|d_{w,a}>\psi\right] \leq \frac{\Gamma \left(\gamma/2+1\right)}{2 P_r \pi^{\gamma/2+1}}$. Therefore,
\begin{comment}
\begin{align}\label{eq:totalex2}
\nonumber \mathbb{E}_{R,W}\left[\left.\frac{m^{\gamma/2}}{{{\sigma}}_w^2}\right|d_{w,a}>\psi\right] &=  \mathbb{E}_{R,W}\left[\left.\frac{m^{\gamma/2}}{\sigma_{w_0}^2+\frac{P_r}{d_{r,w}^\gamma}}\right|d_{w,a}>\psi\right] \\
\nonumber & \leq \frac{m^{\gamma/2}}{P_r} \mathbb{E}_{R,W}\left[\left.d_{r,w}^\gamma \right|d_{w,a}>\psi\right]\\
\nonumber & = \frac{2 m^{\gamma/2+1} \pi}{P_r} \mathbb{E}_{W}\left[\int_{x=0}^{x=\infty} x^{\gamma+1} e^{-m \pi x^2 } dx\right]\\
  &= \frac{\Gamma \left(\gamma/2+1\right)}{2 P_r \pi^{\gamma/2+1}} 
\end{align}
\noindent where $\mathbb{E}_{W}[.]$ is the expectation over all possbile locations of Willie. Thus

\end{comment}
\begin{align} \label{eq:th2eq2}
 \mathbb{E}_{R,W} \left[\left. \mathbb{P}_{FA} + \mathbb{P}_{MD}\right|d_{w,a}>\psi\right] & \geq 1-\frac{c \Gamma \left(\gamma/2+1\right)}{4 \sqrt{2} \psi^{\gamma} P_r \pi^{\gamma/2+1}} ,
\end{align}

\noindent Since $\psi\leq \frac{1}{2}$, $ \mathbb{P}(d_{w,a}>\psi) = 1-\pi \frac{\psi^2}{2}$. The law of total expectation yields
\begin{align} 
&\mathbb{E}_{R,W} [ \mathbb{P}_{FA} + \mathbb{P}_{MD}]  \nonumber \\
& \nonumber\geq  \mathbb{E}_{R,W} \left[\left. \mathbb{P}_{FA} + \mathbb{P}_{MD}\right|d_{w,a}>\psi\right]\; \mathbb{P}(d_{w,a}>\psi)\\
& \nonumber \geq \left( 1-\frac{c \Gamma \left(\gamma/2+1\right)}{4 \sqrt{2} \psi^{\gamma} P_r \pi^{\gamma/2+1}}\right)\left(1-\pi \frac{\psi^2}{2}\right)\\
& \geq \left( 1-\frac{c \Gamma \left(\gamma/2+1\right)}{4 \sqrt{2} \psi^{\gamma} P_r \pi^{\gamma/2+1}}-\pi \frac{\psi^2}{2}\right)\
\end{align}

\noindent Now, first choosing $\psi =\sqrt{\frac{2 \epsilon}{\pi}}$ and then  $c =  \epsilon \left(\frac{\Gamma \left(\gamma/2+1\right)}{4 \sqrt{2}  \psi^{\gamma} P_r \pi^{\gamma/2+1}} \right)^{-1}$, $\mathbb{E}_{R,W} [ \mathbb{P}_{FA} + \mathbb{P}_{MD}] \geq  1-\epsilon$ for any $\epsilon>0$ as long as $P_f = \mathcal{O}\left(\cfrac{m^{\gamma/2}}{\sqrt{n}}\right)$.

Next, we analyze Bob's ML decoder. The law of total expectation yields
\begin{align}\label{eq:EPEBobth2}
\nonumber \mathbb{P}_e &=  \mathbb{E}_{R,W}[\mathbb{P}_e(\sigma_b^2,d_{w,a}^2)]\\ \nonumber 
&\leq \mathbb{E}_{R,W}[\mathbb{P}_e(\sigma_b^2,d_{w,a}^2)|d_{r,b}>  \phi] \\
 & \phantom{=} +  \mathbb{P}\left(d_{r,b}\leq   \phi \right)
\end{align}

\noindent where $d_{r,b}$ is the distance between Bob and the closest friendly node to Willie, and $0<\phi\leq1$. If the rate is set to $R= \frac{\rho}{2} \log_2 \left( 1+ \frac{ c  m^{\gamma/2} }{ 2 \sqrt{n} \left(\sigma_{b_0}^2+\frac{P_r}{\phi^{\gamma}}\right)} \right) $, when $0 < \rho < 1$, by~\eqref{eq:th1bobpe}, the first term on the RHS of Eq.~\eqref{eq:EPEBobth2} is:
\begin{align}
&\mathbb{E}_{R,W}[\mathbb{P}_e(\sigma_b^2,d_{w,a}^2)|d_{r,b}>  \phi]     \leq  \left({1+{ \frac{c(1-\rho)  m^{\gamma/2} \sqrt{n}}{ 4(\sigma_{b_0}^2+\frac{P_r}{\phi^{\gamma}})}}}\right)^{-1}
\label{eq:th2bob1}
\end{align}

\begin{comment}

% why eq:th2bob2} is correct:
\noindent The third term is

\begin{align}\label{eq:th2bob22}
 \mathbb{P}\left(\{d_{r,b}\leq   \phi\} \cap \{d_{w,b}>   \phi\} \right) 
 &\nonumber  \leq \mathbb{P}\left(d_{rw}> d_{w,b}-\phi\right)\\
 & = e^{-\pi m \left(d_{w,b}-\phi\right)^2}
\end{align}
\end{comment}

\noindent Since $m=o\left(n^{{1}/{\gamma}}\right)$, $\lim\limits_{m,n \to \infty} \mathbb{E}_{R,W}[\mathbb{P}_e(\sigma_b^2,d_{w,a}^2)|d_{r,b}>  \phi] = 0$. Now, consider $\mathbb{P}\left(d_{r,b}\leq \phi\right)$. Since $\{d_{r,b} \leq \phi\} \subset \{ \{d_{w,b} \leq 2 \phi \} \cup \{d_{r,w} \geq  \phi \}  \} $, $\mathbb{P} \left(d_{r,b} \leq \phi\right) \leq \mathbb{P} \left(d_{w,b} \leq 2 \phi\right)  + \mathbb{P} \left(d_{r,w} \geq  \phi\right) $. As $m \to \infty$, $\mathbb{P} \left(d_{r,w} \geq  \phi\right) \to 0$ and thus the right side approaches $2 \pi \phi^2$. Thus, setting $\phi = \sqrt{\frac{\zeta}{2 \pi}}$ means $\lim\limits_{m,n \to \infty} \mathbb{P}_e < \zeta$ for any $0< \zeta < 1 $.

Next, we calculate the average number of bits that Bob can receive. Similar to the approach that leads to Eq.~\eqref{eq:thm1bobp3}, we can easily show that $nR \geq  \frac{\sqrt{n} \rho c  m^{\gamma/2} }{ 4 \left(\sigma_{b_0}^2+\frac{P_r}{\phi^{\gamma}}\right) }$. Thus, Bob receives $\mathcal{O}(m^{\gamma/2} \sqrt{n})$ bits in $n$ channel uses.
 \end{proof}

\section{Covert Communication in the Presence a Multiple Collaborating Wardens}
\label{mult_willies}

In this section, we consider case when there are $N_w$ collaborating Willies located independently in the 1 by 1 square.

\begin{thm}

When friendly nodes are distributed such that $m={o}\left(n^{{1}/{\gamma}}\right)$ and $m=\omega(1)$ and 
$N_w=\mathcal{O}\left(m^{\frac{ \gamma}{\gamma+2}}\right)$ collaborating Willies are uniformly and independently distributed over the unit square  shown in Fig.~\ref{fig:SysMod}, Alice can reliably and covertly transmit $  \mathcal{O}\left(\frac{m^{\gamma/2}\sqrt{n}}{N_w^{2+\gamma}}\right) $ bits to Bob over $n$ uses of the channel.
\begin{comment}

Conversely, if Alice attempts to transmit $\omega(\frac{m^{\gamma/2}\sqrt{n}}{N_w^2}) $ bits to Bob over $n$ uses of channel, either there exists a detector that Willie can use to detect her with arbitrarily low sum of error probabilities $\mathbb{P}_{FA} + \mathbb{P}_{MD}$ or Bob cannot decode the message with arbitrarily low probability of error.
\end{comment}
\end{thm}

\begin{proof}

\textbf{Construction:} The codebook construction is same as that in Theorem 1. Analogously to the constructions of Theorems 1 and 2, Alice and Bob's strategy is to turn on the closest friendly node to each Willie and keep all other friendly nodes off, whether Alice is transmitting or not.

\textbf{Analysis:} When Willie applies the optimal hypothesis test, Pinsker's Inequality (Lemma 11.6.1 in \cite{coverthomas}) yields\cite{bash_jsac2013}
\begin{equation}\label{eq:basic2}
\mathbb{P}_{FA} + \mathbb{P}_{MD} \geq 1- \sqrt{\frac{1}{2} \mathcal{D}(\mathbb{P}_1 || \mathbb{P}_0)} .
\end{equation}

Here, $\mathbb{P}_0$ and  $\mathbb{P}_1$ are the joint probability distributions of Willies' channels observations  for the $H_0$ and $H_1$ hypotheses respectively; in other words
\begin{equation}
\mathbb{P}_0 = [\mathbb{P}_0^{(w_1)^T} \mathbb{P}_0^{(w_2)^T} \dots \mathbb{P}_0^{(w_{N_w})^T}]^T
\end{equation}

\begin{equation}
\mathbb{P}_1 = [\mathbb{P}_1^{(w_1)^T} \mathbb{P}_1^{(w_2)^T} \dots \mathbb{P}_1^{(w_{N_w})^T}]^T
\end{equation}

\noindent where $\mathbb{P}_0^{(w_k)} $ is the vector probability distribution of the channel observation of Willie $W_k$ ($ 1 \leq k \leq N_w$) when $H_0$ is true and includes $n$ elements with the same probability distribution $\mathbb{P}_{w_k} = \mathcal{N}(0, \sigma_{w_k}^2) $. In addition, $\mathbb{P}_1^{(w_k)} $ is the channel observation of Willie $W_k$ when $H_1$ is true and includes $n$ elements, each with the same probability distribution $\mathbb{P}_{w_k} = \mathcal{N}(0, \sigma_{w_k}^2+\frac{P_f}{d_{w_k,a}^{\gamma}}) $.

\noindent The relative entropy between two multivariate normal distributions $\mathbb{P}_1$ and  $\mathbb{P}_0$ is given by \cite{klref}:
\begin{align} \label{eq:kl}
 \nonumber \mathcal{D}(\mathbb{P}_1 || \mathbb{P}_0) &= { 1 \over 2 } \left(
  \mathrm{tr} \left( \Sigma_0^{-1} \Sigma_1 \right) + \left( \mu_0 - \mu_1\right)^\top \Sigma_0^{-1} ( \mu_0 - \mu_1 )\right. \\
  &\left. \phantom{=}  - \mathrm{dim}\left(\Sigma_0\right) - \ln \left( { | \Sigma_1 | \over | \Sigma_0 | } \right)  \right)
 \end{align}

\noindent where $\mathrm{tr}(.)$, $|.|$, and $\mathrm{dim}(.)$ denote the trace, determinant and dimension of a square matrix respectively, $\mu_0=0$, $\mu_1=0$ are the mean vectors, and $\Sigma_0$, $\Sigma_1$ are nonsingular covariance matrices of $\mathbb{P}_0$ and  $\mathbb{P}_1$ respectively and are given by
\begin{align}
\Sigma_0 & =S\otimes I_{n \times n}\\
\Sigma_1 & = \big(S + P_f U U^T\big)\otimes I_{n \times n}
\end{align}

\noindent where $ S= \mathrm{diag}(\sigma_{w_1}^2, ... \, ,\sigma_{w_{N_w}}^2)$, $\otimes$ denotes the Kronecker product between two matrices, $I_{n \times n}$ is the identity matrix of size $n$, and $U$ is a column vector of size $N_w$ given by
\begin{equation}
   U= \begin{bmatrix}
    \frac{1}{d_{w_1,a}^{\gamma/2}} & \frac{1}{d_{w_2,a}^{\gamma/2}} & \dots & \frac{1}{d_{w_{N_w},a}^{\gamma/2}} \
  \end{bmatrix} ^{T}
\end{equation}

\noindent Next, we calculate the relative entropy in~\eqref{eq:kl}. The first term of the RHS of~\eqref{eq:kl} is:
\begin{align}
\nonumber \mathrm{tr} \left( \Sigma_0^{-1} \Sigma_1 \right)& = n \sum_{k=1}^{N_w} \frac{1}{\sigma_{w_k}^2}\left( \sigma_{w_k}^2 + \frac{P_f}{d_{w_k,a}^{\gamma}}\right)\\
 & = n N_w+ n \sum_{k=1}^{N_w} \frac{P_f}{d_{w_k,a}^{\gamma} \sigma_{w_k}^2}
\end{align}

\noindent Then,
\begin{align}
\nonumber\left|{\Sigma_0}\right|&=  \left|S\otimes I_{n \times n}\right|\\ \label{eq:kron_sig0}&=\left|S\right|^n | I_{n \times n}|^{N_w} \\ \nonumber &=\left|S\right|^n\\
& = \left(\prod_{k=1}^{N_w} \sigma_{w_k}^2 \right)^n.
\end{align}

\noindent where~\eqref{eq:kron_sig0} is due to the determinant of the kronecker product property presented in~\cite{matrixkron}. Because each of the Willies has non-zero noise variance, $S$ is nonsingular. Therefore,
\begin{align} 
\nonumber \left|{\Sigma_1}\right|& =\left|S + P_f U U^T\right|^n | I_{n \times n}|^{N_w}\\ 
\nonumber &=\left|S + P_f U U^T\right|^n\\
\nonumber  &= \left|S\right|^n \left|I+P_f S^{-1} U U^T \right|^n\\
\label{eq:matrixdetlemma} &= \left|S\right|^n \left(1+P_f  U^T S^{-1} U\right)^n\\
 &= |\Sigma_0| \left(1+\sum_{k=1}^{N_w} \frac{P_f}{d_{w_k,a}^{\gamma} \sigma_{w_k }^2} \right)^n
\end{align}

\noindent where step~\eqref{eq:matrixdetlemma} is due to Lemma 1.1 in \cite{matrixDing}. Therefore, 
\begin{align}
\ln \left( { | \Sigma_1 | \over | \Sigma_0 | } \right) &= n \ln{\left(1+\sum_{k=1}^{N_w} \frac{P_f}{d_{w_k,a}^{\gamma} \sigma_{w_k }^2} \right)}.
\end{align}

\noindent Thus,
\begin{align}
\mathcal{D}(\mathbb{P}_1 || \mathbb{P}_0) = {n \over 2}\left(  \sum_{k=1}^{N_w} \frac{P_f}{d_{w_k,a}^{\gamma} \sigma_{w_k}^2} - \ln{\left(1+\sum_{k=1}^{N_w} \frac{P_f}{d_{w_k,a}^{\gamma} \sigma_{w_k}^2}\right)} \right).
\end{align}

Suppose Alice sets her average symbol power $P_f \leq \frac{c  m^{\gamma/2}}{\sqrt{n} N_w} $ where $c$ is a constant defined later. Since $m=o\left(n^{{1}/{\gamma}}\right)$, for $n$ large enough $\sum_{k=1}^{N_w} \frac{P_f}{d_{w_k,a}^{\gamma} \sigma_{w_k}^2} < 1$. Therefore
\begin{equation}
\label{eq:dformulw}\mathcal{D}(\mathbb{P}_1 || \mathbb{P}_0) \leq {n \over 4}\left( \sum_{k=1}^{N_w} \frac{P_f}{d_{w_k,a}^{\gamma} \sigma_{w_k}^2} \right)^2.
\end{equation}

\noindent as long as $d_{w_k,a} > \kappa$ for all $k$. Assume $Q$ is the event that $d_{w_k,a}>\kappa $ for all $k$ where $0<\kappa<\frac{1}{2}$. By Eqs.~\eqref{eq:basic2} and~\eqref{eq:dformulw}
\begin{align} \label{eq:th4ERW}
&\nonumber \mathbb{E}_{R,W}\left[\left.\mathbb{P}_{FA} + \mathbb{P}_{MD}\right|Q\right] \\
&\nonumber \geq 1- \mathbb{E}_{R,W} \left[\left.  {1 \over 2} \sqrt{\frac{n}{2}} \sum_{k=1}^{N_w} \frac{P_f}{d_{w_k,a}^{\gamma} \sigma_{w_k}^2}\right|Q\right]\\
&\nonumber \geq 1-{ c \over {2 \sqrt{2} N_w }}  \mathbb{E}_{R,W}\left[\left.\sum_{k=1}^{N_w} \frac{m^{\gamma/2}}{d_{w_k,a}^{\gamma} \sigma_{w_k}^2}\right|Q\right]\\
&  \geq 1-{ c  \over {2 \sqrt{2} N_w \kappa^{\gamma}}}  \sum_{k=1}^{N_w} \mathbb{E}_{R,W}\left[\left. \frac{m^{\gamma/2}}{ \sigma_{w_k}^2}\right|Q\right]
\end{align}

\noindent As we obtained in~\eqref{eq:totalex1}, $\mathbb{E}_{R,W}\left[\left.\frac{m^{\gamma/2}}{{{\sigma}}_{w_k}^2}\right|d_{w_k,a}>\kappa\right] \leq \frac{\Gamma \left(\gamma/2+1\right)}{2 P_r \pi^{\gamma/2+1}}$ for all $k$. Therefore,
\begin{align} \label{eq:th4ERW2}
 \mathbb{E}_{R,W}\left[\left.\mathbb{P}_{FA} + \mathbb{P}_{MD}\right|Q\right] &  \geq 1- \frac{c \Gamma \left(\gamma/2+1\right)}{2 \sqrt{2} \kappa^{\gamma} 2 P_r \pi^{\gamma/2+1}}
\end{align}

\noindent Since $\kappa<\frac{1}{2}$, $\mathbb{P}(Q) = \left(1-\frac{\pi \kappa^2}{2}\right)^{N_w}$. Then, the law of total expectation yields
\begin{align} 
\nonumber &\mathbb{E}_{R,W} [ \mathbb{P}_{FA} + \mathbb{P}_{MD}] \\&  \nonumber \geq  \mathbb{E}_{R,W} \left[\left. \mathbb{P}_{FA} + \mathbb{P}_{MD}\right|Q\right]\; \mathbb{P}(Q) \\
&\nonumber  \geq \left(1- \frac{c \Gamma \left(\gamma/2+1\right)}{2 \sqrt{2} \kappa^{\gamma} 2 P_r \pi^{\gamma/2+1}}\right)\left(1-\frac{\pi \kappa^2}{2}\right)^{N_w}\\
&\geq \left(1-\frac{\pi \kappa^2}{2}\right)^{N_w}- \left(\frac{c \Gamma \left(\gamma/2+1\right)}{2 \sqrt{2} \kappa^{\gamma} 2 P_r \pi^{\gamma/2+1}}\right).
\end{align}

\noindent Thus, for any $\epsilon > 0$ and $N_w$,  $\kappa=\sqrt{\frac{2}{\pi}\left(1-\left(1-\frac{\epsilon}{2}\right)^{\frac{1}{N_w}}\right)}$ and $c=\frac{\epsilon}{2}\left(\frac{\Gamma \left(\gamma/2+1\right)}{4 \sqrt{2}  \kappa^{\gamma} P_r \pi^{\gamma/2+1}} \right)^{-1}$ yields  $\mathbb{E}_{R,W} [ \mathbb{P}_{FA} + \mathbb{P}_{MD}] \geq  1-\epsilon$ as long as $P_f = \mathcal{O}\left(\cfrac{m^{\gamma/2}}{\sqrt{n} N_w^{1+\gamma/2}}\right)$.

Next, we analyze Bob's ML decoding error probability over all possible codewords as well as the locations of Willies and closest friendly node to each Willie. Bob's noise power is given by
\begin{align}
\label{eq:th4bob0} \sigma_b^2 & \leq \sigma_{b_0}^2 + \sum_{k=1}^{N_w} \frac{P_r}{d_{r_k,b}^2}
\end{align}
\noindent where $d_{r_k,b}$ is the distance between Bob and the closest friendly node to Willie $W_k$. Suppose $G$ is the event that $d_{r_k,b} > \delta$ for all $k$ where $0\leq\delta\leq1$. Therefore
\begin{align}\label{eq:EPEBobth41}
\mathbb{P}_e&= \mathbb{E}_{R,W}[\mathbb{P}_e\left(\sigma_b^2\right)] \nonumber\\
& \leq \mathbb{E}_{R,W}[\mathbb{P}_e\left(\sigma_b^2\right)|G] +   \mathbb{P}(\bar{G})
\end{align}

\noindent Similar to what we did in ~\eqref{eq:Pebasic}-\eqref{eq:th1bobpe}, if the rate is set to $R= \frac{\rho}{2} \log_2 \bigg( 1+ \frac{c_1 m^{\gamma/2} }{2 N_w^{1+\gamma/2} \sqrt{n} \left(\sigma_{b_0}^2+N_w \frac{P_r}{\delta^\gamma}\right)} \bigg) $, where $c_1 = c N_w^{\gamma/2}$ and $0 < \rho < 1$,  the first term of the RHS of~\eqref{eq:EPEBobth41} is
\begin{align}\label{eq:EPEBobth42}
\nonumber \mathbb{E}_{R,W}[\mathbb{P}_e|G] &\leq  \mathbb{E}_{R,W}[\mathbb{P}_e|d_{r_1,b}=\dots=d_{r_{N_w},b}=\delta]\\
& =  \left({1+{ \frac{c_1 (1-\rho)  m^{\gamma/2} \sqrt{n}}{4 N_w^{1+\gamma/2}\left(\sigma_{b_0}^2 + \frac{N_w P_r}{\delta^{\gamma}}\right)}}}\right)^{-1}.
     \end{align}

\noindent For a given $N_w$, choose $\delta = \frac{1}{2}\sqrt{\frac{2 \zeta}{\pi N_w}}$. If we set $N_w=o\left(m^{\frac{ \gamma}{2+\gamma}}\right)$, $\lim\limits_{m,N_w \to \infty} \mathbb{E}_{R,W}[\mathbb{P}_e|G]=0$.  Consider $\mathbb{P}\left(\bar{G}\right)$
\begin{align}\label{eq:EPEBobth43}
\nonumber \mathbb{P}\left(\bar{G} \right)&  =  \mathbb{P} \left(\bigcup_{k=1}^{N_w} d_{r_k,b}\leq  \delta \right)\\
    & \nonumber \leq  \sum\limits_{k=1}^{N_w}   \mathbb{P} \left( d_{r_k,b}\leq  \delta \right)\\  
     &\nonumber = N_w    \mathbb{P} \left( d_{r_1,b}\leq  \delta \right)\\
     & \nonumber \leq N_w \left(\mathbb{P}\left(d_{w_1,b} \leq 2 \delta \right)  +\mathbb{P}\left(d_{r_1,w_1} \geq  \delta \right) \right)\\
     & \leq N_w \left(\pi \frac{(2 \delta)^2}{2}+e^{-m \pi \delta^2}\right)
\end{align}

\noindent  Then, $N_w e^{-m \pi \delta^2}\to 0$ as $m\to \infty$, and  $\lim\limits_{m,n \to \infty} \mathbb{P}_e < \zeta$ for any $0<\zeta<1$.

Now, we calculate the number of bits that Bob receives. Similar to the approach that leads to Eq.~\eqref{eq:thm1bobp3}, we can easily show that $nR \geq  \frac{\sqrt{n} \rho c_1  m^{\gamma/2} }{ 4 N_w^{1+\gamma/2} \left(\sigma_{b_0}^2+N_w \frac{P_r}{\delta^{\gamma}}\right) }$. Since $\delta = \frac{1}{2}\sqrt{\frac{2 \zeta}{\pi N_w}}$,  for $m, n, N_w$ large enough
\begin{align}
 \nonumber nR   \geq \frac{\sqrt{n} \rho c_1  m^{\gamma/2} { \left(  \frac{\zeta}{2 \pi}\right)}^{\gamma/2}}{4 N_w^{2+\gamma} {P_r}}  \end{align}

\noindent Therefore, Bob receives $\mathcal{O}\left(\frac{m^{\gamma/2} \sqrt{n}}{N_w^{2+\gamma}}\right)$ bits in $n$ channel uses.

\end{proof}

\section{Conclusion}

In this paper, we have considered the first step in establishing low probability of detection (LPD) 
communications in a network scenario.  We established that Alice can transmit $O(m^{\gamma/2}\sqrt{n})$ bits
reliably to the desired recipient Bob in $n$ channel uses without detection by an adversary Willie if
randomly distributed system nodes of density $m$ are available to aid in jamming Willie; conversely, no higher
covert rate is possible.  The presence of multiple collaborating wardens inhibits communication in two
separate ways - increasing the effective signal-to-noise ratio (SNR) at the wardens' decision point, and requiring more interference which inhibits Bob's ability to reliably decode the message.  Future work consists  of embedding the results of this single-hop formulation into large multi-hop covert networks.

\bibliographystyle{ieeetr}
%\bibliography{mycitation}

\end{document}